\documentclass{amsart}
\usepackage{amsmath}
\usepackage{amssymb}
\usepackage{mathrsfs}
\usepackage{amsthm}
\usepackage{array}

\newtheorem{thm}{Theorem}
\newtheorem{dfn}{Definition}
\newtheorem{lem}{Lemma}
\newtheorem{prop}{Proposition}
\newtheorem{cor}{Corollary}

\usepackage{txfonts}

\def\min#1{\mathbb{R}^{#1,1}}
\def\minc#1{\mathbb{C}^{#1,1}}
\def\poin#1{I\hspace{-.1em}O_{#1,1}}
\def\so#1{S\hspace{-.1em}O_{#1}}
\def\F{{}^\sharp F}
\def\G{{}^\sharp G}
\def\del#1{\dfrac{\partial}{\partial #1}}
\def\os{\mathscr{M}}

\title[Integrability of Cohomogeneity-one strings]
      {Complete Integrability of Cohomogeneity-one strings in $\mathbb{R}^{n,1}$
      and Canonical Form of Killing Vector Algebra}
\author[D. Ida]{Daisuke Ida} 
\address{Department of Physics, Gakushuin University, Tokyo 171-8588, Japan}
\begin{document}
\maketitle

\begin{abstract}
  The equation of motion for comohogeneity-one Nambu-Goto strings
  in flat space $\mathbb{R}^{n,1}$ has been investigated.
  We first classify possible forms of the Killing vector fields
  in $\mathbb{R}^{n,1}$ after appropriate action of the Poincar\'e group.
  Then, all possible forms of the  Hamiltonian
  for the cohomogeneity-one Nambu-Goto strings are determined.
  It has been shown that the system always has the
  maximum number of functionally independent,
  pair-wise commuting conserved quantities,
  i.e. it is completely integrable.

  We have also determined all the possible coordinate forms of the 
  Killing vector basis for the 2-dimensional non-commutative Lie algebra.
\end{abstract}
\section{Introduction}

According to the standard models of elementary particles,
the quantum vacuum of certain scalar field depends on
the environmental temperature.
As a result, the scalar field
suffers the phase transition with the cosmic temperature dropping.
Such a phase transition does not occur uniformly in the universe,
but different phases of quantum vacua form a domain structure.
Hence, cosmic membranes and strings unavoidably appear at the borders
of locally homogeneous regions as topological defects in
quantum-vacuum configuration.
These topological defects would play an important role in the formation of
various inhomogeneous structures in the universe.
With this background, the dynamics of relativistic extended objects
such as strings or membranes has been investigated
from various aspects.

In a certain situations, the dynamics of relativistic strings or membranes can
approximately be  described by the Nambu-Goto action, which is the main subject of the present paper.
The Nambu-Goto strings and membranes are just the zero-mean-curvature timelike surface in the spacetime, so that
they might be interesting objects of study also from mathematical viewpoint.

Since it is  a  system with the infinite degrees of freedom,
the dynamics of the Nambu-Goto string or membrane is much more complicated
as compared with the particle dynamics.
However, the simplification of the equation of motion
occurs when the background spacetime has a symmetry,
and the string or membrane respects it
\cite{FSZH89,KF08,HI18}.
In particular, when the spacetime admits a Killing vector field, we can consider
Nambu-Goto strings those world-sheets are invariant under the action of the
isometry.
Such  strings are called cohomogeneity-one strings.

It is well-known that general form of the 
Nambu-Goto strings in flat spacetime $\min{n}$ are just given by
 the $n+1$ harmonic functions.
This is because the equation of motion for the string is reduced to
$n+1$ linear wave equations when  isothermal coordinates
are taken as the world-sheet coordinates.
This approach however would not be suitable for finding the cohomogeneity-one
strings.

Our purpose here is to understand the entity of the cohomogeneity-one
strings in $\min{n}$.
Since there are maximum number of Killing vector fields in $\min{n}$,
a systematic classification of a Killing vector field is required
for this problem.
The classification of the cohomogeneity-one strings
in $\min{3}$ has been performed by \cite{KKI10}.
In the flat spacetime, a Killing vector field can
always be written in inertial coordinate system
as $\xi_\nu=F_{\mu\nu}x^\mu+f_\nu$,
in terms of a constant alternative matrix $F_{\mu\nu}$
and a constant vector $f_\nu$.
Any pair $(F_{\mu\nu},f_\nu)$ corresponds to a Killing vector field
in a certain inertial coordinate system, and it undergoes
the action of Poincar\'e group under the change of the
inertial coordinate system.
The action of Poincar\'e group defines the equivalence relation
in the space of pairs $(F_{\mu\nu},f_\nu)$.
We first
enumerate the
equivalence classes of pairs $(F_{\mu\nu},f_\nu)$,
which give the canonical classification of Killing vector fields
in $\min{n}$.

It has also been shown in \cite{KKI10} that the equation of motion
for cohomogeneity-one strings in $\min{3}$ are all completely integrable,
so that every solution to it is given by quadrature.
We will extend this result to the case of $\min{n}$.

In the course of the analysis of the Killing vector fields in $\min{n}$,
we also obtained the canonical classification
of the 2-dimensional non-abelian Lie algebra of Killing vector fields,
which will also be reported here.

The organization of this paper is as follows:
In Sec. 2, the general properties of the
Killing vector field in $\min{n}$ is described.
In Sec. 3, the
constant  2-form $F_{\mu\nu}$, which characterizes the
Killing vector field,
in $\min{n}$ under the action of the Lorentz group are investigated.
In Sec. 4, the canonical classification of the Killing vector field
in $\min{n}$ is given.
In Sec. 5, the complete integrability of the cohomogeneith-one Nambu-Goto
strings is shown.
In Sec. 6, the canonical form of the representation of the
2-dimensional noncommutative Lie algebra of Killing vector fields
in $\min{n}$ is obtained.
In Sec. 7, we summerize our work.

\section{Killing vector fields in Minkowski spacetime}
We first review elementary properties of Killing vectors in
the Minkowski spacetime.

We denote the $(n+1)$-dimensional Minkowski spacetime by $(\min{n},\eta)$.
Let $(x^\mu)_{\mu=0,1,\dots,n}$ be the inertial coordinate system for $\min{n}$,
where the metric takes the form
\begin{align*}
  \eta=-(dx^0)^2+\sum_{i=1}^n(dx^i)^2.
\end{align*}
A Killing vector field on $\min{n}$, which is subject to the partial
differential equation
\begin{align*}
  \partial_\mu \xi_\nu+\partial_\nu\xi_\mu=0,
\end{align*}
is generally written as
\begin{align*}
  \xi=F_{\mu\nu}x^\mu dx^\nu+f_\nu dx^\nu,
\end{align*}
in terms of the constant alternating matrix $F_{\mu\nu}$
and the constant vector $f_\nu$.

Our first task is to classify all the canonical form of $\xi$,
under the action of the Poincar\'e group (or the inhomogeneous Lorentz group)
$\poin{n}$.

The Poincar\'e  transformation is given by the coordinate transformation
\begin{align*}
  x'^\mu=\varLambda^\mu{}_\nu x^\nu+c^\mu
\end{align*}  
between inertial coordinate systems.
Accordintly, the Killing vectof field transforms as
\begin{align*}
  \xi&=F'_{\mu\nu}x'^\mu dx'^\nu+f'_\nu dx'^\nu,\\
  F'_{\mu\nu}&=\varLambda_\mu{}^\lambda\varLambda_\nu{}^\rho F_{\lambda\rho},\\
  f'_\nu&=\varLambda_\nu{}^\lambda f_\lambda -F'_{\mu\nu}c^\mu.
\end{align*}
This can be interpreted as that
the element $g_{(\varLambda,c)}\in \poin{n}$ of the Poincar\'e group
acts on $F_{\mu\nu}$ and $f_\nu$ as
\begin{align*}
  g_{(\varLambda,c)}:(F_{\mu\nu},f_\nu)
  \longmapsto \left(
  \varLambda_\mu{}^\lambda\varLambda_\nu{}^\rho F_{\lambda\rho},
  \varLambda_\nu{}^\lambda
  f_\lambda-\varLambda_\mu{}^\lambda\varLambda_\nu{}^\rho F_{\lambda\rho}c^\mu
  \right),
\end{align*}
which gives a representation of $\poin{n}$.
Hence, in particular,  $F_{\mu\nu}$ transforms as a Lorentz tensor.

\section{Classification of constant 2-forms}
Each Killing vector is characterized by a constant alternative matrix $F_{\mu\nu}$
and a constant vector $f_\nu$, and these undergo the Poincar\'e transformation.
We want to enumerate the Killing vectors in the Minkowski spacetime
$\min{n}$
up to the Poincar\'e transformation.
In particular, the constant alternative matrix $F_{\mu\nu}$ transforms as
a 2-form. In the following, we give a canonical form of $F_{\mu\nu}$ under the action
of the homogeneous Lorentz transformation.

We first show the following useful fact.

\begin{prop}
  Every constant 2-form $F$ in $\min{n}$ can be made into the
form
\begin{align*}
  F=\sum_{k=0}^{n-1}u_k dx^kdx^{k+1},
\end{align*}
under the action of the special orthogonal group $\so{n}$,
which preserves $\sum_{i=1}^n (x^i)^2$.
\end{prop}

\begin{proof}
  A constant 2-form $F$ can be written as
  \begin{align*}
    F&=dx^0\wedge \omega_1+F_1,\\
\omega_1&=\sum_{i=1}^n (\omega_1)_i dx^i,\\
F_1&=\sum_{1\le i<j\le n}(F_1)_{ij}dx^idx^j.
  \end{align*}
There is an $\so{n}$ action preserving $\sum_{i=1}^n(x^i)^2$,
such that $\omega_1$ takes the form
\begin{align*}
  \omega_1=u_0 dx^1.
\end{align*}
Under such a transformation, $F$ becomes
\begin{align*}
  F=u_0dx^0dx^1+F_1,
\end{align*}
where $F_1$ is generated only by $dx^1,\dots,dx^n$.

In a similar manner,  $F_1$ can be written,
under the appropriate $\so{n-1}$ transformation
preserving $\sum_{k=2}^n(x^k)^2$,
in the form
\begin{align*}
  F_1=u_1dx^1dx^2+F_2,
\end{align*}
where $F_2$ is generated by $dx^2,\dots,dx^n$.

We can continue  recursively
until we finally get
\begin{align*}
  F=u_0dx^0dx^1+u_1dx^1dx^2+\dots+u_{n-1}dx^{n-1}dx^n.
\end{align*}
\end{proof}

Hence all the information of a constant 2-form is encoded 
in the array $(u_0,\dots,u_{n-1})$, however, 
it is still not canonical in the sense
that diffrent arrays may give equivalent 2-forms.
In the following, we remove such redundancy along similar lines 
as in the anti-de Sitter case \cite{BHTZ,HP97,MR04}.

We denote by
$\F$ the linear operator on 
the complexified vector space $\minc{n}=\mathbb{C}\otimes_{\mathbb{R}}\min{n}$, 
or explicitly $(\F v)^\mu =F^\mu{}_\nu v^\nu$ for $v\in \minc{n}$.

The following  Lemmas are useful for the classification of the
constant 2-form in $\min{n}$.
The argument here is mostly taken from Ruse \cite{Rus36}.
\begin{lem}
  If $\lambda$ is an eigenvalue of the linear operator $\F$,
which corresponds to the constant 2-form $F$, then
$-\lambda$, $\overline{\lambda}$, and $-\overline{\lambda}$
are  eigenvalues of $\F$.
\end{lem}

\begin{proof}
By regarding the 2-form $F$ as a real, alternating matrix,
the statement of the Lemma follows from
\begin{align*}
0=  \operatorname{det}(\lambda \eta-F)
&= \operatorname{det}(\lambda \eta+F)\\
&= \overline{\operatorname{det}(\overline{\lambda} \eta-F)}\\
&= \overline{\operatorname{det}(\overline{\lambda} \eta+F)}.
\end{align*}
\end{proof}

\begin{lem}\label{lem:ortho}
Let $\lambda$ and $\mu$ be eigenvalues of $\F$, such that $\lambda+\mu\ne 0$,
and let   $v$, $w$ be eigenvector corresponding to $\lambda$, $\mu$, respectively. Then, $v$ and $w$ are mutually orthogonal with respect to 
the bilinear form $\eta$. 
\end{lem}
\begin{proof}
  The statement of the Lemma follows from
  \begin{align*}
    \mu \eta(v,w)=\eta(v,\F w)=\eta((\F)^Tv,w)=-\eta(\F v,w)
=-\lambda \eta(v,w)\\
\therefore (\lambda+\mu)\eta(v,w)=0.
  \end{align*}
\end{proof}

\begin{cor}\label{cor:null}
  An eigenvector corresponding to a nonzero eigenvalue of $\F$
is a null vector.
\end{cor}

Note that the null eigenvector here may be a complex null vector
such as $\partial_1+i\partial_2$.

\begin{lem}\label{lem:eigenvalue}
Each eigenvalue of $\F$ is real or pure imaginary.
\end{lem}
\begin{proof}
  Suppose that $\F$ has an eigenvalue $\lambda=a+ib$ $(a,b\ne 0)$.
Let $v$ be an eigenvector corresponding to $\lambda$.
By Corollary. \ref{cor:null}, $v$ is a null vector.
By reality of $\F$, $\overline v$ is an eigenvector corresponding to
the eigenvalue $\overline\lambda$.
Since $\lambda+\overline{\lambda}\ne 0$, 
$v$ and $\overline v$ are mutually orthogonal by
Lemma \ref{lem:ortho}.
This is possible only when $v$ is a real null vector.

However, it is impossible since, for the real eigenvector $v$,  
$v$ should simaltaneously be
an eigenvector corresponding to $\lambda$ and $\overline\lambda$,
which leads to a contradiction.
\end{proof}

Note that this Lemma also is specific to the Lorentzian signature case.
When the underlying space is $\mathbb{R}^{n,m}$ $(n,m\ge 2)$,
$\F$ may have  general complex eigenvalues (compare with the anti-de Sitter
case \cite{MR04}).

\begin{lem}\label{lem:realnull}
  Every eigenvector corresponding to a nonzero real eigenvalue of $\F$
is a real null vector.
\end{lem}
\begin{proof}
  Let $\lambda$ be a nonzero real eigenvalue of $\F$, and let $v$ be
an eigenvector corresponding to $\lambda$.
By reality of $\F$, $\overline v$ also is an eigenvector corresponding to
$\lambda$.
By Lemma \ref{lem:ortho}, $v$ and $\overline v$ are
mutually orthogonal with respect to $\eta$.
This occurs only when $v$ is a real null vector.
\end{proof}

\begin{lem}\label{lem:realeigenvalue}
If $\F$ has a nonzero real eigenvalue $a$,
$\F$ does not 
have nonzero real eigenvalues other than $a$ or $-a$.
\end{lem}

\begin{proof}
  Let $a$ and $a'$ ($a\ne a'$) be nonzero real eigenvalues
of $\F$, and let $v$ and $w$ be eigenvectors corresponding to
$a$ and $a'$, respectively. 
By Lemma \ref{lem:realnull}, $v$ and $w$ are real null vectors.
This implies that $a'=-a$, since otherwise, by Lemma \ref{lem:ortho},
 $v$ and $w$ are
mutually orthogonal with respect to $\eta$, which is impossible
for linearly independent real null vectors $v$ and $w$.
\end{proof}

\begin{lem}
  An eigenvector corresponding to a nonzero pure imaginary eigenvalue
of $\F$ is a complex null vector.
\end{lem}
\begin{proof}
  Let $ib\ne 0$  be a pure imaginary eigenvalue of $\F$, 
and let $u$ be an corresponding eigenvector.
By Corollary \ref{cor:null}, $u$ is a null vector.
Since $\overline u$ is an eigenvector corresponding to
the different eigenvalue $-ib$ of $\F$, $u$ and $\overline u$ should be 
linearly independent.
Hence $u$ must be a complex null vector.
\end{proof}

The canonical forms of the constant 2-form $F$ can be obtained
by choosing coordinate basis vectors respecting eigenvectors
of $\F$.
When $\F$ has a nonzero real eigenvalue, this procedure is particularly simple.

\begin{lem}\label{lem:nonzeroreal}
  If $\F$ has a nonzero real eigenvalue $a$,
$F$ can be brought into the form
\begin{align*}
F&=adx^0dx^1,
\end{align*}
or
\begin{align*}
  F&=adx^0dx^1+\sum_{k=1}^r b_k dx^{2k}dx^{2k+1},~~~
(b_1\ge b_2\ge\dots\ge b_r>0)
\end{align*}
by  an action of the Lorentz group $O_{n,1}$.
\end{lem}
\begin{proof}
  Let $v$ be a eigenvector corresponding to the nonzero eigenvalue $a$.
By Corollary \ref{cor:null}, $v$ is a real null vector.
We can choose an inertial coordinate system, under which $v$ takes the form
\begin{align*}
  v=\dfrac{1}{\sqrt{2}}(\partial_0-\partial_1).
\end{align*}
The 2-form $F$ generally takes the form
\begin{align*}
  F=fdx^0dx^1+\sum_{i=2}^n g_i(dx^0+dx^1)dx^i+\sum_{i=2}^n h_i(dx^0-dx^1)dx^i
+\sum_{2\le i<j\le n}H_{ij}dx^idx^j.
\end{align*}
However, since $v$ is a eigenvector corresponding to $a$,
$f=a$ and $h_i=0$ $(i=2,3,\dots, n)$ holds, so that
we have
\begin{align*}
  F=adx^0dx^1+\sum_{i=2}^n g_i(dx^0+dx^1)dx^i
+\sum_{2\le i<j\le n}H_{ij}dx^idx^j.
\end{align*}

Under the Lorentz transformation
\begin{align*}
x'^0&= \left(1+\dfrac{\beta^2}{2}\right) x^0+\dfrac{\beta^2}{2}x^1
+\sum_{k=2}^n\beta_k x^k,\\
x'^1&= -\dfrac{\beta^2}{2}x^0+\left(1-\dfrac{\beta^2}{2}\right)x^1
-\sum_{k=2}^n\beta_k x^k,\\
x'^i&= x^i+\beta_i(x^0+x^2),&(i=2,3,\dots,n)
\end{align*}
where $\beta^2=\sum_{k=2}^n(\beta_k)^2$, $F$ transforms as
\begin{align*}
F\longmapsto adx^0dx^1+\sum_{  2\le i<j\le n}H_{ij}dx^idx^j
+\sum_{i=2}^n\left[\sum_{j=2}^n(a\delta_{ij}+H_{ij})\beta_j+g_i\right](dx^0+dx^1)dx^j.
\end{align*}
Noting that $(a\delta_{ij}+H_{ij})$ constitutes a regular matrix,
the third term above can be set to zero, by appropriately chosing
$\beta_j$.

Finally, by an $O_{n-1}$ transformation preserving $\sum_{k=2}^n(x^k)^2$,
the second term above can be made into the block diagonal form as
\begin{align*}
  \sum_{  2\le i<j\le n}H_{ij}dx^idx^j
&\longmapsto \sum_{i=1}^r b_idx^{2i}dx^{2i+1},
\end{align*}
such that $b_1\ge b_2\ge\dots\ge b_r>0$ holds, whenever it is nonzero.
\end{proof}

Then, the following Lemma  immediately follows. 
\begin{lem}\label{lem:surjective}
If  $\F\in \operatorname{End}(\minc{n})$ 
is surjective, then $n$ is odd and $F$ can be written as
\begin{align*}
  F=adx^0dx^1+\sum_{i=1}^{(n-1)/2}b_idx^{2i}dx^{2i+1},
\end{align*}
in an appropriate inertial coordinate system.
\end{lem}
\begin{proof}
  Since $\operatorname{ker}\F=\{0\}$, every eigenvalue of $\F$ is
  nonzero, real or pure imaginary number by Lemma \ref{lem:eigenvalue}.

  Consider the characteristic polynomial
  \begin{align*}
    p(x)=\operatorname{det}(xI-\F).
  \end{align*}
  Since this is a real polynomial, the number of pure imaginary roots counting multiplicity  is even.
  According to Lemma \ref{lem:realeigenvalue},
  the number of nonzero real eigenvalues is 0 or 2.
  Hence the degree of $p(x)$, which is $n+1$,  must be even.

  If $p(x)$ does not have no real roots, it is written as
  \begin{align*}
    p(x)=\prod_{i=1}^{(n+1)/2}(x^2+b_i^2),
  \end{align*}
  where every $b_i$ is real. However, it is impossible since
  \begin{align*}
    p(0)=\prod_{i=1}^{(n+1)/2}b_i^2>0,
  \end{align*}
  while
  \begin{align*}
    p(0)=\operatorname{det}\F=
    -\operatorname{det}F\le 0.
  \end{align*}
  Hence, $\F$ has nonzero real roots.
 Then, the statement of the Lemma immediately follows from
  Lemma \ref{lem:nonzeroreal}.
\end{proof}

Now we are in a position to classify the canonical forms of the
constant 2-form $\F$, according to the causal nature of
$\operatorname{ker}\F$.
Since $\F$ is real, when $u\in \operatorname{ker}\F$, 
its complex conjugate also belongs to $\operatorname{ker}\F$.
This implies that $\operatorname{ker}\F$ is the complexfication
of a subspace of $\min{n}$.
We first classify 
a subspace of $\min{n}$ according to its causal nature.
\begin{dfn}[subspace of $\min{n}$]
  A nonzero subspace of $\min{n}$
is classified into the following types.
\begin{enumerate}
\item A subspace generated by mutually orthogonal,
a real timelike vector and $m$ real spacelike vectors
is called an $(m+1)$-dimensional timelike subspace, where $m=0,1,\dots,n$.
\item A subspace generated by  mutually orthogonal,
$m$ real spacelike vectors
is called an $m$-dimensional spacelike subspace, where $m=1,2,\dots,n$.
\item A subspace generated by mutually orthogonal,
 a real null vector and
$m$ real spacelike vectors
is called an $(m+1)$-dimensional null subspace, where $m=0,1,\dots,n-1$.
\end{enumerate}
\end{dfn}

According to the causal nature of $\operatorname{ker}\F$,
the constant 2-form $F$ can be classified as follows.

\begin{prop}\label{prop:spacelike}
  If $\operatorname{ker}\F$ is a complexified spacelike subspace of $\min{n}$,
a constant 2-form $F$ is brought into the form
\begin{align*}
  F&=adx^0dx^1
&(a>0)
\end{align*}
or
\begin{align*}
  F&=adx^0dx^1+\sum_{i=1}^r b_idx^{2i}dx^{2i+1},
&(a>0,b_1\ge b_2\ge\dots\ge b_r>0)
\end{align*}
under an action of the Lorentz group $O_{n,1}$.
\end{prop}

\begin{proof}
  We can take an inertial coordinate system for $\min{n}$, 
such that
$\partial_{m+2},\partial_{m+3},\dots,\partial_n$
generate $\operatorname{ker}\F$.
Then, $\partial_0,\partial_1,\dots,\partial_{m+1}$
generates $\operatorname{Im}\F$,
which is isomorphic to the quotient vector space $\minc{n}/\operatorname{ker}\F$,
and $\operatorname{Im}\F$ is identified with a complexified vector space $\minc{m}$.
In this coordinate system, $F$ takes the form
\begin{align*}
  F=\sum_{0\le i,j\le m+1}F_{ij}dx^i dx^j,
\end{align*}
and
$\F$ is regarded as a surjective linear operator on $\operatorname{Im}\F$.
By Lemma \ref{lem:surjective},
the complex dimension of $\operatorname{Im}\F$ should be even,  or $m=2r$ ($r=0,1,\dots$),
and $F$ can be brought into the form
\begin{align*}
  F=adx^0dx^1,
\end{align*}
when $r=0$, or otherwise
\begin{align*}
  F&=adx^0dx^1+\sum_{i=1}^r b_idx^{2i}dx^{2i+1},&
(b_1\ge b_2\ge\dots\ge b_r>0)
\end{align*}
by an $O_{2r+1,1}$ transformation. 
Furthermore, we can always make $a$ positive.
\end{proof}

\begin{prop}\label{prop:timelike}
  If $\operatorname{ker}\F$ is a complexified
timelike subspace of $\min{n}$,
a constant 2-form $F$ is brought into the form
\begin{align*}
  F&=\sum_{i=1}^rb_idx^{2i-1}dx^{2i},&
(b_1\ge b_2\ge\dots\ge b_r>0)
\end{align*}
under an action of the Lorentz group $O_{n,1}$.
\end{prop}
\begin{proof}
  We can always choose an inertial coordinate system
such that $\partial_0,\partial_{m+1},\partial_{m+2},\dots,\partial_n$
generate $\operatorname{ker}\F$ 
(of course, we take $\operatorname{ker}\F
=\mathbb{C}\partial_0$ if the kernel is 1-dimensional). 
Then, $F$ takes the form
\begin{align*}
  F=\sum_{1\le \mu<\nu\le m}F_{\mu\nu}dx^\mu dx^\nu.
\end{align*}
Hence it can be brought into the form
\begin{align*}
  F&=\sum_{i=1}^rb_idx^{2i-1}dx^{2i},&
(b_1\ge b_2\ge\dots\ge b_r>0)
\end{align*}
by an $O_m$ transformation.
\end{proof}

\begin{prop}\label{prop:null}
  If $\operatorname{ker}\F$ is a complexified null subspace of $\min{n}$,
a constant 2-form $F$ is brought into the form
\begin{align*}
  F&=dx^0dx^1+dx^1dx^2
\end{align*}
or
\begin{align*}
  F&=dx^0dx^1+dx^1dx^2+\sum_{i=1}^rb_idx^{2i+1}dx^{2i+2},&
(b_1\ge b_2\ge\dots\ge b_r>0)
\end{align*}
under an action of the Lorentz group $O_{n,1}$.
\end{prop}

\begin{proof}
  Since $\operatorname{ker}\F$ is a complexified null subspace of $\min{n}$,
$\F$ has an real null eigenvector corresponding to the eigenvalue 0.
We choose an inertial coordinate system for $\min{n}$, such that
$\ell=\partial_0+\partial_2$ is the null eigenvector.

In this coordinate system,
$F$ generally takes the form
\begin{align*}
  F=\sum_{i\ne 0,2}g_i(dx^0-dx^2)dx^i
+\dfrac{1}{2}\sum_{i,j\ne 0,2}H_{ij}dx^idx^j,
\end{align*}
where $h_{ij}=-h_{ji}$.
This is the  general form of the 2-form, such that $\ell\in \operatorname{ker}\F$.
Furthermore, in order for $\operatorname{ker}\F$ to be the complexified
null subspace of $\min{n}$,
the 1-form $g=\sum_{k\ne 0,2}g_kdx^k$ should not be contained in
 $\operatorname{Im}H$, where $H=(1/2)\sum_{i,j\ne 0,2}H_{ij}dx^idx^j$
is regarded as a linear operator from the tangent vector space
into the space of 1-forms. The reason is as follows.

 Suppose that there is a vector $v=\sum_{k\ne 0,2}v^k\partial_k$ such that
$g_i=\sum_{j\ne 0,2}H_{ij}v^j$ holds for $i\ne 0,2$.
Then, it is easily confirmed that
$u=\partial_0+\sum_{k\ne 0,2}v^k\partial_k$ belongs to $\operatorname{ker}\F$.
On the other hand, since $\operatorname{ker}\F\cap \min{n}$ is
a null hypersurface in $\min{n}$, every pair of 
vectors in $\operatorname{ker}\F$ should be mutually orthogonal.
Hence it leads to a contradiction, since $\ell$, $u\in \operatorname{ker}\F$,
while $\eta(\ell,u)=-1\ne 0$.

In particular, $H$ is not surjective as a linear operator on
the $(x^1,x^3,x^4,\dots,x^n)$-plane, so that by $O_{n-1}$ transformation
on this plane, it may be brought into the form
\begin{align*}
  H&=\dfrac{1}{2}\sum_{i=1}^rb_i dx^{2i+1}dx^{2i+2},
&(b_1\ge b_2\ge\dots\ge b_r>0)
\end{align*}
whenever it is not zero.

Under the Lorentz transformation
\begin{align*}
  x'^0&=\left(1+\dfrac{\beta^2}{2}\right)x^0-\dfrac{\beta^2}{2}x^2
+\sum_{i\ne 0,2}\beta_i x^i,\\
  x'^2&=\dfrac{\beta^2}{2}x^0+\left(1-\dfrac{\beta^2}{2}\right)x^2
+\sum_{i\ne 0,2}\beta_i x^i,\\
x'^i&=x^i+\beta_i(x^0-x^2),&(i\ne 0,2)
\end{align*}
where $\beta^2=\sum_{k\ne 0,2}(\beta_k)^2$, 
$H$ is invariant, while $g$ transforms as
  \begin{align*}
g\longmapsto g=\sum_{i\ne 0,2} \left(g_i- \sum_{j\ne 0,2}H_{ij}  \beta_j   \right) dx^i.
  \end{align*}
In other words, $g$ can be shifted by any vector in $\operatorname{Im}H$.
Hence by appropriately choosing $\beta_i$, $g$ can be brought into the
form
\begin{align*}
  g=g_1dx^1+\sum_{i=2r+3}^n g_i dx^i,
\end{align*}
where the second term in the right hand side may be absent
when $2r+2=n$.
Note that this is not zero, since otherwise $g\in \operatorname{Im}H$.
Then, by an $\so{n-2r-1}$ (or $\mathbb{Z}_2$ if $2r+2=n$)
transformation  preserving
$(x^1)^2+\sum_{i=2r+3}^n(x^i)^2$, it becomes
\begin{align*}
 g&\longmapsto \| g\|x^1.
&\left(\|g\|=\sqrt{(g_1)^2+\sum_{i=2r+3}^n(g_i)^2}\right)
\end{align*}

At present,  the 2-form $F$ can be written in the form
\begin{align*}
  F=\|g\|(dx^0-dx^2)dx^1
+\sum_{i=1}^rb_idx^{2i+1}dx^{2i+2}.
\end{align*}

Finally, by the Lorentz transformation
  \begin{align*}
x'^0-x'^2&=\|g\|(x^0-x^2),\\    
x'^0+x'^2&=\|g\|^{-1}(x^0+x^2),\\    
x'^i&=x^i,&(i\ne 0,2)
  \end{align*}
$\|g\|$ in $F$ can be made 1. 
So, finally we get
  \begin{align*}
F=dx^0dx^1-dx^1dx^2,
  \end{align*}
when $H_{ij}=0$, or otherwise
  \begin{align*}
F&=dx^0dx^1-dx^1dx^2+\sum_{i=1}^r b_idx^{2i+1}dx^{2i+2}.
&(b_1\ge b_2\ge\dots\ge b_r>0)
  \end{align*}
\end{proof}

From Propositions \ref{prop:spacelike}, \ref{prop:timelike}, and
\ref{prop:null},
we obtain the classfication Theorem for  constant 2-forms.
\begin{thm}\label{thm:F}
  Let $F$ be a constant 2-form on $\min{n}$.
Under an action of the Lorentz group $O_{n,1}$,
$F$ can be brought into one of following forms
\begin{enumerate}
  \renewcommand{\labelenumi}{(\alph{enumi})}
\item $u=(0,\dots,0)$
  \begin{align*}
  F&=0
  \end{align*}

\item $u=(0,b_1,0,b_2,\dots,0,b_r,0,0,\dots,0)$
\begin{align*}
F&=\sum_{i=1}^r b_idx^{2i-1}dx^{2i}
&(b_1\ge b_2\ge\dots\ge b_r>0)
\end{align*}

\item $u=(a,0,0,\dots,0)$
  \begin{align*}
    F&=adx^0dx^1
&(a>0)
  \end{align*}

\item $u=(a,0,b_1,0,b_2,0,\dots,b_r,0,0,\dots,0)$
  \begin{align*}
    F&=adx^0dx^1+\sum_{i=1}^r b_idx^{2i}dx^{2i+1}
&(a>0,b_1\ge b_2\ge\dots\ge b_r>0)
  \end{align*}

\item $u=(1,1,0,0,\dots,0)$
  \begin{align*}
F&=dx^0dx^1+dx^1dx^2    
  \end{align*}

\item $u=(1,1,0,b_1,0,b_2,\dots,0,b_r,0,0,\dots,0)$
  \begin{align*}
    F&=dx^0dx^1+dx^1dx^2    
+\sum_{i=1}^r b_i dx^{2i+1}dx^{2i}
&(b_1\ge b_2\ge\dots\ge b_r>0)
  \end{align*}
\end{enumerate}
and it gives a canonical classification of constant 2-forms.  
\end{thm}

 
\section{Canonical classification of Killing vector field}
Here, we enumerate all canonical form of the Killing vector field
in $\min{n}$.
The result seems to be already known by some authors, see e.g. \cite{FS01},
whereas the complete classification and its reasoning
do not seem to appear in the literature.
Here, we give the canonical classification for self-containedness.

A Killing vector field can be brought into the form
\begin{align*}
  \xi=F_{\mu\nu}x^\mu dx^\nu+f_\nu dx^\nu,
\end{align*}
where $F_{\mu\nu}$ is one of canonical forms in Theorem \ref{thm:F}.
In order to find the canonical form of $\xi$, we study the
transformation of $f_\nu$ under the subgroup $G_F$ of the  Poincar\'e group
$\poin{n}$ that preserves $F_{\mu\nu}$.

\begin{dfn}
For given canonical 2-form $F$,  we denote by $G_F$ the subgroup
of $\poin{n}$
that leaves $F$ invariant, i.e.,
\begin{align*}
  G_F=\left\{g_{(\varLambda,c)}\in\poin{n}\bigl|
\varLambda_\mu{}^\lambda \varLambda_\nu{}^\rho F_{\lambda\rho}=F_{\mu\nu}\right\}.
\end{align*}
\end{dfn}

Under the action of $g_{(\varLambda,c)}\in G_F$, the constant 1-form $f$ transforms as
\begin{align*}
g_{(\Lambda,c)} (f)&=\left(\varLambda_\mu{}^\nu f_\nu+F_{\mu\nu}c^\nu\right) dx^\mu.
\end{align*}
In particular, we may always shift $f$ by a 1-form $Fc$ belonging to $\operatorname{Im}F$, where $F$ is regarded as a linear operator from the 
real tangent space to the space of real 1-forms.

By Theorem \ref{thm:F}, the matrix representation of $F$
in canonical form is
a block diagonal matrix consisting of 
the following types of blocks:
\begin{itemize}
\item $S$-block
  \begin{align*}
S&=\left(
  \begin{array}{cc}
    0&s\\
-s&0
  \end{array}\right),
\end{align*}

\item $N$-block
\begin{align*}
N=\left(
  \begin{array}{ccc}
    0&1&0\\
-1&0&1\\
0&-1&0
  \end{array}\right),
  \end{align*}

\item $O$-block

 \begin{align*}
O&=\left(\,
\begin{array}{cccc}
0&0&\dots&0\\
0&0&\dots&0\\
\vdots&\vdots&\ddots&\vdots\\
0&0&\dots&0
\end{array}
\,\right)    
  \end{align*}
\end{itemize}
where $s>0$.

For given $F$, let $W_S$ be the image of the $S$-block part, i.e.,
if 
\begin{align*}
F=adx^0dx^1+b_1dx^2dx^3+\dots+b_r dx^{2r}dx^{2r+1},  
\end{align*}
then
\begin{align*}
  W_S=\bigoplus_{i=0}^{2r+1} \mathbb{R}dx^i,
\end{align*}
and if 
\begin{align*}
F=dx^0dx^1+dx^1dx^2+b_1dx^3dx^4+\dots+b_r dx^{2r+1}dx^{2r+2},  
\end{align*}
then
\begin{align*}
W_S=\bigoplus_{i=3}^{2r+2} \mathbb{R}dx^i,
\end{align*}
and so on, and let $W_N$ be its orthogonal complement of $W_S$
with respect to $\eta$.
Hence any 1-form $f$ can be written  as
$f=f_S+f_N$, $(W_S\in W_S, f_N\in W_N)$
according to this direct sum decomposition.

Under the translation $x'^\mu=x^\mu+c^\mu$,
the constant 1-form $f$ transforms as $f\mapsto f+Fc$, 
so that $f_S$ can always be made zero.
Hence the problem reduce to find out all
the canonical forms of $f_N$ under the
action of $G_F$.

The $N$-block and $O$-block part of $F$ can be regarded as
the linear operator $F_N:V_N\to W_N$, where
\begin{align*}
  V_N=\eta^{-1}(W_N)=\left\{v\in \min{n}\bigl|\eta(v)\in W_N\right\},
\end{align*}
and $F_N$ is the restriction of $F$ on $V_N$.
Then, the following cases appear.

\begin{itemize}
\item $V_N$ is $(x^s,x^{s+1},\dots, x^n)$ space, where $s=2,3,\dots,n$:

$F_N$ is zero on $V_N$, and $f_N$ can be brought into the form
$f_n dx^n$ $(f_n\ge 0)$ by an $O_{n-s}$ transformation.

\item $V_N$ is $(x^0,x^s,x^{s+1},\dots,x^n)$ space, where $s=1,3,4,\dots,n$,
and $F_N=O$:

By a Lorentz transformation over $W_N$, $f_N$ can be made
into the form
$f_0dx^0$ $(f_0\le 0$), $f_n dx^n$ $(f_n>0)$, or $-dx^0+dx^n$,
according to its causal property.

\item $V_N$ is $(x^0,x^1,x^2)$ space and $F_N=N$:

$f_N$ can be made $f_0 dx^0$  by a translation.
Then, $f_0$ can be made nonpositive by $x^0\mapsto -x^0$,
if necessary.

\item $V_N$ is $(x^0,x^1,x^2,x^s,x^{s+1},\dots,x^n)$ space,
where $s=3,4,\dots,n$,
and $F_N=N\oplus O$:

Firstly, any $f_N$ can be broughot into the form
$f_0dx^0+f_1dx^1+f_2 dx^2+f_n dx^n$ $(f_n\ge 0)$
by $O_{n-s-1}$ transformation over $(x^s,\dots,x^n)$ space.
Next, it can be made
into the form $f_0dx^0+f_ndx^n$ by a translation.
If $f_0\ne 0$, $f_n$ can be made zero, by a Lorentz transfomation
over
$(x^0,x^1,x^2,x^n)$ space. So that the final form of $f_N$ is
$f_0 dx^0$ $(f_0\le 0)$ or $f_n dx^n$ $(f_n>0)$.
\end{itemize}

We are now in a position to state the classification Theorem 
for Killing vector fields.
\begin{thm}\label{thm:killing}
  Each Killing vector field $\xi$ in $\min{n}$ is brought into 
  one of the following forms
  by the $\poin{n}$ transformation.
\begin{enumerate}
  \renewcommand{\labelenumi}{(\alph{enumi})}
\item
  $n\ge 1$
  \begin{align*}
    \xi=f_0dx^0,
    f_0< 0
  \end{align*}

\item
  $n\ge 1$
  \begin{align*}
    \xi=f_ndx^n,
    f_n>0
  \end{align*}
\item
  $n\ge 1$
  \begin{align*}
    \xi=-dx^0+dx^n,
  \end{align*}
  
\item
  $n\ge 1$
  \begin{align*}
    \xi=a(x^0dx^1-x^1dx^0),
    a> 0
  \end{align*}

\item
  $n\ge 2$, $r\le n/2$
    \begin{align*}
    \xi=\sum_{i=1}^r b_i (x^{2i-1}dx^{2i}-x^{2i}dx^{2i-1})+f_0 dx^0,\\
    b_1\ge b_2\ge\dots\ge b_r>0,
f_0\le 0
  \end{align*}

  \item
    $n\ge 2$
  \begin{align*}
    \xi=a(x^0dx^1-x^1dx^0)+f_ndx^n,
    a> 0,
    f_n> 0
  \end{align*}

\item
  $n\ge 2$
  \begin{align*}
    \xi=(x^0-x^2)dx^1-x^1(dx^0-dx^2)+f_0dx^0,
    f_0\le 0
  \end{align*}
\item
  $n\ge 3$,  $r< n/2$
  \begin{align*}
    \xi=\sum_{i=1}^r b_i (x^{2i-1}dx^{2i}-x^{2i}dx^{2i-1})+f_n dx^n,\\
    b_1\ge b_2\ge\dots\ge b_r>0,
f_n> 0
  \end{align*}

\item
  $n\ge 3$, $r< n/2$
  \begin{align*}
    \xi=\sum_{i=1}^r b_i (x^{2i-1}dx^{2i}-x^{2i}dx^{2i-1})-dx^0+dx^n,\\
    b_1\ge b_2\ge\dots\ge b_r>0
  \end{align*}

\item
  $n\ge 3$, $r\le (n-1)/2$
  \begin{align*}
  \xi=a(x^0dx^1-x^1dx^0)+\sum_{i=1}^r b_i (x^{2i}dx^{2i+1}-x^{2i+1}dx^{2i}),\\
  a>0,
  b_1\ge b_2\ge\dots\ge b_r>0
\end{align*}

\item
  $n\ge 3$
  \begin{align*}
    \xi=(x^0-x^2)dx^1-x^1(dx^0-dx^2)+f_ndx^n,
    f_n>0
  \end{align*}
\item
  $n\ge 4$,  $r<(n-1)/2$
\begin{align*}
  \xi=a(x^0dx^1-x^1dx^0)+\sum_{i=1}^r b_i (x^{2i}dx^{2i+1}-x^{2i+1}dx^{2i})
  +f_n dx^n,\\
  a>0,
  b_1\ge b_2\ge\dots\ge b_r>0,
f_n> 0
\end{align*}

\item
  $n\ge 4$, $r\le (n/2)-1$
     \begin{align*}
       \xi=(x^0-x^2)dx^1-x^1(dx^0-dx^2)+
       \sum_{i=1}^r b_i (x^{2i+1}dx^{2i+2} -x^{2i+2}dx^{2i+1} )
       +       f_0dx^0,\\
          b_1\ge b_2\ge\dots \ge b_r>0,
    f_0\le 0
  \end{align*}

   \item
     $n\ge 5$, $r<(n/2)-1$
     \begin{align*}
       \xi=(x^0-x^2)dx^1-x^1(dx^0-dx^2)+
       \sum_{i=1}^r b_i (x^{2i+1}dx^{2i+2} -x^{2i+2}dx^{2i+1} )
       +       f_ndx^n,\\
          b_1\ge b_2\ge\dots \ge b_r>0,
    f_n> 0
  \end{align*}

\end{enumerate}

\end{thm}

\section{Complete integrability of cohomogeneity-1 string system}

As an application of Theorem \ref{thm:killing}, we show the
complete integrability of the cohomogeneity-1 Nambu-Goto string system
in flat spacetime.

We first describe the cohomogeneity-1 Nambu-Goto system.
The world sheet of a string is
a 2-dimensional timelike surface in $\min{n}$,
which is written as
\begin{align*}
  x^\mu=\varphi^\mu(s^0,s^1),~~~(\mu=0,1,\dots,n)
\end{align*}
where $\varphi^\mu$'s are smooth functions, and the parameters $(s^0,s^1)$
are called  world-sheet coordinates.
If the world-sheet of the string has the zero mean curvature, 
it is called the Nambu-Goto string.

The induced metric on the world-sheet
\begin{align*}
  G_{AB}=\eta_{\mu\nu}
  \dfrac{\partial \varphi^\mu}{\partial s^A}
  \dfrac{\partial \varphi^\mu}{\partial s^B}
\end{align*}
is called the world-sheet metric.
The equation of motion for the Nambu-Goto strings is
given by
\begin{align}
  \label{eq:flat-NGeq}
    \dfrac{\partial}{\partial s^A}
  \left(\sqrt{-{\rm det}~G}G^{AB}
  \dfrac{\partial \varphi^\mu}{\partial s^B}\right)=0,
\end{align}
which is just the requirement that $\varphi^\mu$'s are harmonic functions
on the world-sheet.

Choosing a thermal coordinate on the  world-sheet,
this reduces to the wave equation in $\min{1}$,
so that we have ordinary plane-wave solutions for this system.

However, we here consider
the cohomogeneity-1 string, that respects the
spacetime symmetry in $\min{n}$ generated by a Killing vector field.
In this case, the equation of motion reduce to a geodesic equation on
a certain Riemannian, or Lorentzian $n$-manifold
that is the space of the Killing orbits.

In order to find the ignorable coordinate
associated with the Killing vector field $\xi$,
we first solve the coupled ordinary differential equations
\begin{align*}
  \dfrac{dx^\mu(\tau)}{d\tau}=\xi^\mu(x(\tau)).
\end{align*}
The solution may be obtained in the form
\begin{align*}
  x^\mu=x^\mu(\tau,y^1,y^2,\dots,y^n),
\end{align*}
where $(y^a)_{a=1,\dots,n}$ parametrize integral curves of $\xi^\sharp(:=\xi^\mu\partial_\mu)$.
This just gives the coordinate transfomation.
In the $(\tau,y^a)$ coordinate system, the Killing vector field
is simply given by
\begin{align*}
  \xi^\sharp=\del{\tau}.
\end{align*}
Hereafter, we denote by $g$, instead of $\eta$,  
the spacetime metric on $\min{n}$,
when we are concerned with noninertial coordinates. 

The metric in this coordinate system is written as
\begin{align*}
  g=X\left(ds-\sum_{a=1}^nW_ady^a\right)^2+\sum_{a,b=1}^n h_{ab}dy^ady^b,
\end{align*}
where $X$, $W_a$, and $h_{ab}$ depend only on $(y^a)_{a=1,\dots,n}$.
In particular, $X$ is the square of $\xi^\sharp$,
i.e., $X=g(\xi^\sharp,\xi^\sharp)$,
and $h$ is regarded as the natural metric on 
the space of Killing orbits, which aquires the structure of
a differentiable $n$-manifold.
We denote by $\os$ the space of Killing orbits.

The cohomogeneity-1 string is described as
\begin{align*}
s&=\tau,\\
y^a&=\varphi^a(\sigma),~~~(a=1,2,\dots,n)
\end{align*}
in terms of the world-sheet coordinate $(\tau,\sigma)$.
Then, it is known \cite{FSZH89} that the equation of motion (\ref{eq:flat-NGeq}) becomes
the geodesic equation for $\varphi^a(\sigma)$ on the Riemannian, or the Lorentzian $n$-manifold
$(\os,\gamma)$, where the
metric $\gamma$ is given by
\begin{align*}
\gamma_{ab}&=|X|h_{ab}.
\end{align*}

We show the complete integrability of this system
utilizing the conformal trick developed in \cite{KKI10,MHKI17},
which is brilliantly useful for finding pair-wise Poisson commuting
conserved quantities for the geodesic Hamiltonian system on $(\os,\gamma)$.

The Lagrangian for geodesic in $(\os,\gamma)$ is
given by
\begin{align*}
  L(y,\dot y)=-\sqrt{\pm \gamma_{ab}\dot y^a\dot y^b},
\end{align*}
where the double-sign in the square root takes plus-sign if 
$(\os,\gamma)$ is Riemannian
(i.e. if $X<0$), otherwise
minus-sign if it is Lorentzian ($X>0$).
This is a singular Lagrangian system, since it leads to
the constraint
\begin{align*}
\varPhi:=  \gamma^{ab}p_ap_b\mp 1\approx 0.
\end{align*}
Then the Hamiltonian becomes zero, 
and the system is described by the total Hamiltonian
\begin{align*}
  H(y,p,\lambda)=\lambda \varPhi.
\end{align*}
It immediately follows that $\varPhi$ is a first-class constraint,
and the Lagrange multiplier $\lambda$ remains undetermined.
This reflects the invariance of the system
under the reparametrization of the
Hamiltonian time parameter, which in the present case is $\sigma$.

We can fix this ambiguity of time-parametrization simply by putting
\begin{align*}
  \lambda=\dfrac{|X|}{2},
\end{align*}
that corresponds to the time-reparametrization $|X|d\sigma'=\lambda d\sigma$.
After this gauge fixation, the Hamiltonian  becomes
\begin{align}\label{eq:geod-os}
  H(y,p)=\dfrac{1}{2}(h^{ab}p_ap_b+X),
\end{align}
that takes the same form as the particle system with potential function $X/2$.

However, $h^{ab}$ that is the inverse matrix of $h_{ab}$ is just 
the $(a,b)$ component of $g^{\mu\nu}$ that is the inverse matrix of $g^{\mu\nu}$,
so that we can also write as
\begin{align*}
 H(y,p)=\dfrac{1}{2}\left(\sum_{\mu,\nu=1}^n g^{\mu\nu}p_\mu p_\nu+X\right).
\end{align*}
This Hamiltonian system is equivalent with 
the geodesic Hamiltonian
\begin{align*}
 H(s,y,p_s,p)=\dfrac{1}{2}\left(\sum_{\mu,\nu=0}^n g^{\mu\nu}p_\mu p_\nu+X\right),
\end{align*}
with the initial condition $p_s=0$. 
Since $p_s=\xi^\mu p_\mu$ is a conserved quantity, we can simply put $p_s=0$
at the Hamiltonian level.

Further simplification occurs when we perform the point transformation
from $(y^\mu)_{\mu=0,\dots,n}=(s,y^a)$ to the inertial coordinate $(x^\mu)_{\mu=0,\dots,n}$,
that leads to
\begin{align}\label{eq:particle-flat}
  H(x,P)=\dfrac{1}{2}(\eta^{\mu\nu}P_\mu P_\nu+X),
\end{align}
where
the point transformation of the momentum is given by
\begin{align*}
  P_\mu=\dfrac{\partial y^\nu}{\partial x^\mu}p_\nu,
\end{align*}
that is just the transformation law for the cotangent vector.
In particular, the restriction $p_s=0$ is simply written as
\begin{align*}
  \xi^\mu P_\mu=0.
\end{align*}
Thus, the geodesic Hamiltonian system in $(\os,\gamma)$ has been translated
into the particle system in $(\min{n},\eta)$ with the potential
function $X/2$ with the initial condition $\xi^\mu P_\mu=0$.

In order to show the complete integrability of the geodesic Hamiltonian system
defined by (\ref{eq:geod-os}), we have only to find the
$n+1$ pair-wise Poisson commuting conserved quantities for the Hamiltonian
system (\ref{eq:particle-flat}) those are functionally independent.
Among such, the Hamiltonian $H$ and $p_s=\xi^\mu P_\mu$ are already in hand.
Hence our task is to find $n-1$ more pair-wise Poisson 
commuting conserved quantities those
commute with $p_s$.

From Theorem \ref{thm:killing}, a Killing vector in $\min{n}$
can be written as
\begin{align*}
  \xi=F_{\mu\nu}x^\mu dx^\nu+f_\nu dx^\nu,
\end{align*}
where the matrix representation of $F_{\mu\nu}$ 
is a direct sum of $N$, $S$, $O$-blocks.
This globally determines the orthogonal direct sum decomposition
of the space of differential 1-forms on $\min{n}$
\begin{align*}
\varOmega(\min{n})= \varOmega_N\oplus \varOmega_S\oplus \varOmega_O,
\end{align*}
in an obvious manner.
Explicitly, the orthogonal direct sum decomposition of $\xi$,
\begin{align*}
  \xi=\xi_N+\xi_S+\xi_O,
  ~~~\xi_N\in\varOmega_M,
  ~~~\xi_S\in\varOmega_S,
  ~~~\xi_O\in\varOmega_O
\end{align*}
can generally be written as
\begin{align*}
  \xi_N&=(f_0-x^1)dx^0+(x^0-x^2)dx^1+x^1dx^2,\\
  \xi_S&=\sum_{i=1}^r b_i (x^{2i+1}dx^{2i+2}-x^{2i+2}dx^{2i+1}),\\
  \xi_O&=f_ndx^n,
\end{align*}
when $N$ is present, 
\begin{align*}
  \xi_N&=0,\\
  \xi_S&=a(x^0dx^1-x^1dx^0)+\sum_{i=1}^r b_i (x^{2i}dx^{2i+1}-x^{2i+1}dx^{2i}),\\
    \xi_O&=f_ndx^n,
\end{align*}
when $N$ is absent and $\varOmega_S$ is Lorentzian, or
\begin{align*}
  \xi_N&=0,\\
  \xi_S&=\sum_{i=1}^r b_i (x^{2i-1}dx^{2i}-x^{2i}dx^{2i-1}),\\
  \xi_O&=f_0dx^0+f_ndx^n,
\end{align*}
when $N$ is absent and $\varOmega_S$ is Riemannian.
The canonical momentum 1-form $P=P_\mu dx^\mu$ is also decomposed as
\begin{align*}
  P=P_N+P_S+P_O.~~~
  (P_N\in \varOmega_N,~~~
  P_S\in \varOmega_S,~~~
  P_O\in \varOmega_O)
\end{align*}

Then, the Hamiltonian (\ref{eq:particle-flat}) is accordingly decomposed
into three parts as 
\begin{align*}
  H&=H_N+H_S+H_O.
\end{align*}
Here, the $N$-Hamiltonian is defined by
\begin{align*}
  H_N&:=\dfrac{1}{2}\left[\eta^{\mu\nu}(P_N)_\mu (P_N)_\nu+g({}^\sharp\xi_N,{}^\sharp\xi_N)\right]\\
  & =\dfrac{1}{2}\left[-(P_0)^2+(P_1)^2+(P_2)^2
    -(f_0-x^1)^2+(x^0-x^2)^2+(x^1)^2\right],
\end{align*}
if $N$-block is present in the matrix representation of $F_{\mu\nu}$,
and otherwise $H_N=0$.

The $S$-Hamiltonian is similarly defined as
\begin{align*}
  H_S&:=\dfrac{1}{2}\left[\eta^{\mu\nu}(P_S)_\mu (P_S)_\nu
    +g({}^\sharp\xi_S,{}^\sharp\xi_S)\right].
\end{align*}
It is
\begin{align*}
  H_S&=\dfrac{1}{2}\left[-(P_0)^2+(P_1)^2+a^2(x^0)^2-a^2(x^1)^2\right]\\
&  +\dfrac{1}{2}\sum_i\left[(P_{2i})^2+(P_{2i+1})^2
  +(b_i)^2(x^{2i})^2+  (b_i)^2(x^{2i+1})^2\right],
\end{align*}
when $\varOmega_S$ is Lorentzian, and
\begin{align*}
  H_S&=
  \dfrac{1}{2}\sum_i \left[(P_{2i-1})^2+(P_{2i})^2
  +(b_i)^2(x^{2i-1})^2+  (b_i)^2(x^{2i})^2\right],
\end{align*}
when $\varOmega_S$ is Riemannian.

Finally, the $O$-Hamiltonian is defined as
\begin{align*}
  H_O&:=\dfrac{1}{2}\left[\eta^{\mu\nu}(P_O)_\mu (P_O)_\nu
       +g({}^\sharp\xi_O,{}^\sharp\xi_O)\right],
\end{align*}
and it  generally has the form
\begin{align*}
  H_O=\dfrac{1}{2}\left[\sum_\mu \eta^{\mu\mu} (P_\mu)^2-(f_0)^2+(f_n)^2\right].
\end{align*}

These $N$, $S$, $O$-Hamiltonians describe mutually
independ dynamical systems.
With this observation,
the complete integrability of our Hamiltonian system is readily understood.

\begin{thm}
  The Hamiltonian system defined by (\ref{eq:particle-flat})
  is completely integrable.
\end{thm}
\begin{proof}
  We show that $N$, $S$, and $O$-Hamiltonian has as many conserved quantities
  as the size of $N$, $S$, and $O$-block of $F_{\mu\nu}$, respectively.

Firstly, the $N$-Hamiltonian has three conserved quantities
  \begin{align*}
    C_N&:=H_N,\\
    D_N&:=(x^1-f_0)P_0+(x^0-x^2)P_1+x^1P_2,\\
    E_N&:=P_0+P_2.
  \end{align*}

  Secondly, the $S$-Hamiltonian has as many conserved quantities
  as the size of $S$-block, of the following forms
  \begin{align*}
    C_S^0&:=\dfrac{1}{2}\left[-(P_0)^2+(P_1)^2+a^2(x^0)^2-a^2(x^1)^2\right],\\
    D_S^0&:=x^0P_1+x^1P_0,\\
    C_S^i&:=\dfrac{1}{2}\left[(P_i)^2+(P_{i+1})^2+(b_i)^2(x^i)^2+(b_i)^2(x^{i+1})^2\right],\\
        D_S^i&:=x^iP_{i+1}-x^{i+1}P_i.
  \end{align*}

  Finally, the $O$-Hamiltonian has as many conserved quantities as
  the size of $O$-block, which are simply 
  \begin{align*}
    F_O^i:=P_i.
  \end{align*}

  These constitute $n+1$
  functionally independent,
  mutually Poisson commutative conserved quantities.
\end{proof}

In fact, the Hamiltonian system (\ref{eq:particle-flat}) can
generally  be solved in terms of elementary functions.
Then, the solution to the geodesic Hamiltonian system (\ref{eq:geod-os})
is obtained simply the point transformation
of the solution  subject to $\xi^\mu P_\mu=0$, that is just the
coordinate transformation into the stationary coordinate system.

\section{Canonical representation of noncommutative algebra of Killing field}

Here we consider the minimal noncommutative Lie algebra
of Killing vector fields
given by
\begin{align*}
  \left[\xi^\sharp,\eta^\sharp\right]=\xi^\sharp.
\end{align*}
We show that this algebra is realized only by a limited class of
Killing vector pairs.
The result shown in the following will be useful e.g. in analysis
of Nambu-Goto membranes those are invariant under the action
of a 2-dimensional isometry group.

Let $\xi$ and $\eta$ be written as
\begin{align*}
  \xi&=F_{\mu\nu}x^\mu dx^\nu+f_\nu dx^\nu,\\
  \eta&=G_{\mu\nu}x^\mu dx^\nu+g_\nu dx^\nu,
\end{align*}
respectively.
Then, the condition $[\xi^\sharp,\eta^\sharp]=\xi^\sharp$ is equivalent
  to
\begin{align*}
  \left[\F,\G\right]&:=\F\G-\G\F=\F,\\
  (\G+I)f^\sharp&=\F g^\sharp.
\end{align*}

Let us introduce a several notions for a special
type of 2-forms, which turns out to be useful.
\begin{dfn}[superdiagonal array]
  For the following form of the constant 2-form
  \begin{align*}
    F=u_0dx^0dx^1+u_1dx^1dx^2+\dots+u_{n-1}dx^{n-1}dx^n,
  \end{align*}
  the array
  \begin{align*}
    (u_0,u_1,\dots,u_{n-1})
  \end{align*}
  is called the superdiagonal array of $F$.
\end{dfn}

\begin{dfn}[isolated nonzero element]
 A nonzero element of the superdiagonal array of $F$
 is called an isolated nonzero element,
 if it is not next to a nonzero element.
  \end{dfn}

With these terminology, we can state the following Lemma.
\begin{lem}\label{lem:FG}
  Let a pair of constant 2-forms $F$, $G$ be subject to
  the commutation relation
  \begin{align*}
    \left[\F,\G\right]=\F,
  \end{align*}
  and let $F$ be written by a superdiagonal array.
  Then, the superdiagonal array of $F$ does not have an isolated
  nonzero element.
\end{lem}

\begin{proof}
  Let $F$ be written in terms of a superdiagonal array
  $(u_0,u_1,\dots,u_{n-1})$.

Assume that $u_0$ is an isolated nonzero element, i.e.,
  that $u_0\ne 0$ and $u_1=0$ hold.
  Then, $(0,1)$ component of $\left[\F,\G\right]$ becomes
  \begin{align*}
    \left[\F,\G\right]^0{}_1=-G^0{}_2F^2{}_1=G^0{}_2 u_1=0,
  \end{align*}
  while $(\F)^0{}_1=-u_0\ne 0$.
  Hence, $\left[\F,\G\right]=\F$ cannot hold.

Next, assume that $u_i$ is an isolated nonzero element
  for some $i$ $(2\le i\le n-2)$, i.e., that
  $u_i\ne 0$ and $u_{i\pm 1}=0$ hold.
  Then, it holds
  \begin{align*}
    \left[\F,\G\right]^i{}_{i+1}=F^i{}_{i-1}G^{i-1}{}_{i+1}
    -G^i{}_{i+2} F^{i+2}{}_{i+1}
    =-u_{i-1}G^{i-1}{}_{i+1}+u_{i+1}G^i{}_{i+2}=0,
  \end{align*}
  while $(\F)^i{}_{i+1}=u_{i}\ne 0$ holds.
  Hence, $\left[\F,\G\right]=\F$ does not hold again.

  Finally, assume that $u_{n-1}$ is an isolated nonzero element,
  i.e., $u_{n-2}=0$ and $u_{n-1}\ne 0$ hold.
  Then, it holds
  \begin{align*}
    \left[\F,\G\right]^{n-1}{}_n
    =F^{n-1}{}_{n-2}G^{n-2}{}_n=-u_{n-2}G^{n-2}{}_n=0,
  \end{align*}
  while $(\F)^{n-1}{}_n=u_{n-1}\ne 0$.
  Hence, $\left[\F,\G\right]=\F$ does not hold.

  These show that the superdiagonal array of $F$ does not contain
  an isolated nonzero element, otherwize
  $\left[\F,\G\right]=\F$ is not fulfilled.
\end{proof}

The Theorem \ref{thm:F}
combined with this Lemma leads to the following Proposition.

\begin{prop}\label{prop:rankF}
  Let $\xi$ and $\eta$ be Killing fields subject to
  $\left[\xi^\sharp,\eta^\sharp\right]=\xi^\sharp$, and let $\xi$ be
  written as
  \begin{align*}
    \xi=F_{\mu\nu}x^\mu dx^\nu+f_\nu dx^\nu.
  \end{align*}
  Then, $\F$ is a nilpotent matrix.
\end{prop}

\begin{proof}
  By a Lorentz transformation, $F$ can be brought into one of standard forms
  shown in Theorem \ref{thm:F}, which are all written by  superdiagonal arrays.
  By Lemma \ref{lem:FG}, the superdiagonal array of
  $F$ does not have an isolated nonzero element.
  Hence, the only possibility is
  (a) $F=0$, or (e) $F=dx^0dx^1+dx^1dx^2$.
  In both cases, the matrix $\F$ is nilpotent.
\end{proof}

This Proposition
greatly restricts the form of the Killing vectors $\xi$ and $\eta$.
The next Lemma determines the form of $d\eta$ in the case that $d\xi\ne 0$.
\begin{lem}\label{lem:FGF}
  A constant 2-form $G$ that satisfies
  \begin{align*}
    \left[\F,\G\right]&=\F
  \end{align*}
  for
  \begin{align*}
    F=dx^0dx^1+dx^1dx^2
  \end{align*}
  can be brought into the form
  \begin{align*}
    G=dx^0dx^2,
  \end{align*}
  or
  \begin{align*}
    G&=dx^0dx^2+\sum_{i=1}^rb_{2i+1}dx^{2i+1}dx^{2i+2}
    & (b_3\ge b_5\ge\dots\ge b_{2r+1}>0)
  \end{align*}
  by a Lorentz transformation that preserves the form of $F$.
\end{lem}

\begin{proof}
The general form of  a constant 2-form $G$
subject to  $\left[\F,\G\right]=\F$ is given by
\begin{align*}
  G&=dx^0dx^2+h_{01}(dx^0-dx^2)dx^1\\
  &+\sum_{k=3}^n h_{0k} (dx^0-dx^2)dx^k
+\sum_{3\le i<j\le n}h_{ij}dx^idx^j.
\end{align*}
By a Lorentz transformation
\begin{align*}
  x^0&=\left(1+\dfrac{\alpha^2}{2}\right)x'^0-\alpha x'^1-\dfrac{\alpha^2}{2}x'^2,\\
  x^1&=x'^1-\alpha(x'^0-x'^2),\\
  x^2&=\dfrac{\alpha^2}{2}x'^0-\alpha x'^1+\left(1-\dfrac{\alpha^2}{2}\right)x'^2,\\
  x^k&=x'^k,&(k=3,4,\dots,n)
\end{align*}
that leaves $F$ invariant,
the 2-form $G$ transforms to
\begin{align*}
  G&=dx'^0dx'^2+(h_{01}-\alpha)(dx'^0-dx'^2)dx'^1\\
  &+\sum_{k=3}^n h_{0k} (dx'^0-dx'^2)dx'^k
  +\sum_{3\le i<j\le n}h_{ij}dx'^idx'^j.
\end{align*}
By taking $\alpha=h_{01}$,
it becomes
\begin{align*}
  G=dx'^0dx'^2+\sum_{k=3}^n h_{0k} (dx'^0-dx'^2)dx'^k
  +\sum_{3\le i<j\le n}h_{ij}dx'^idx'^j.
\end{align*}

By a further Lorentz transformation
\begin{align*}
  x^0&=\left(1+\dfrac{\beta^2}{2}\right)x'^0-\dfrac{\beta^2}{2}x'^2-
  \sum_{k=3}^n\beta_kx'^k,\\
  x^1&=x'^1,\\
  x^2&=\dfrac{\beta^2}{2}x'^0+\left(1-\dfrac{\beta^2}{2}\right)x'^2
  -\sum_{k=3}^n \beta_kx'^k,\\
  x^k&=x'^k-\beta_k(x'^0-x'^2),&(k=3,4,\dots,n)
\end{align*}
that preserves $F$,
the 2-form $G$ becomes
\begin{align*}
  G&=dx''^0dx''^2+\sum_{3\le i<j\le n}h_{ij}dx''^idx''^j\\
  &-\sum_{i=3}^n[(\delta_{ij}-h_{ij})\beta_j-h_{0i}](dx''^0-dx''^2)dx''^k,
\end{align*}
where $h_{ij}$ is regarded as constituting an alternating matrix.

Since $(\delta_{ij}-h_{ij})$'s constitute a invertible matrix,
the last term can be made zero
by appropriately choosing $\beta_j$.

When $h_{ij}=0$, $G$ is written as $G=dx''^0dx''^2$.
Otherwise, by an $O_{n-1}$ transformation
that preserves $(x''^3)^2+\dots+(x''^n)^2$,
the second term can be brought into the form
\begin{align*}
  h_{ij}dx''^idx''^j&=\sum_{i=1}^rb_{2i+1}dx'''^{2i+1}dx'''^{2i+2},
 &(b_3\ge b_5\ge\dots\ge b_{2r+1}>0)
\end{align*}
where  $2r$ is the matrix rank of $h_{ij}$.
\end{proof}

The following Lemma gives a general form of the noncommutative Killing vector pair,
when $F=dx^0dx^1+dx^1dx^2$.
\begin{lem}\label{lem:FN}
  Let $\xi$ and $\eta$ be Killing vector fields
  subject to $\left[\xi^\sharp,\eta^\sharp\right]=\xi^\sharp$.
  which are written as
  \begin{align*}
    \xi_\nu&=F_{\mu\nu}x^\mu dx^\nu+f_\nu dx^\nu,\\
    \eta_\nu&=G_{\mu\nu}x^\mu dx^\nu+g_\nu dx^\nu,
  \end{align*}
in terms of
\begin{align*}
  F&=dx^0dx^1+dx^1dx^2,
\end{align*}
and
\begin{align*}
  G&=dx^0dx^2,
\end{align*}
or
\begin{align*}
     G&=dx^0dx^2+\sum_{i=1}^rb_{2i+1}dx^{2i+1}dx^{2i+2},
\end{align*}
where $b_{2i+1}\ne 0$ $(i=1,\dots,r)$.
Then, by a Poincar\'e transformation
that simaltaneously leaves $F$ and $G$ invariant,
we can set
  \begin{align*}
f&=    f_\nu dx^\nu=0,\\
g&=    g_\nu dx^\nu=qdx^n.
  \end{align*}
\end{lem}

\begin{proof}
  Recall that $g$ transforms as
  \begin{align*}
    g'_\nu=g_\nu-G_{\mu\nu}c^\mu,
  \end{align*}
  according to the translation
  \begin{align*}
    x'^\mu=x^\mu+c^\mu.
  \end{align*}
  Hence, by taking
  \begin{align*}
   & (c^0,c^1,c^2,c^3,c^4,\dots,c^{2r+1},c^{2r+2},c^{2r+3},\dots,c^n)\\
    =&\left(g_2,0,-g_0,\dfrac{g_4}{b_3},-\dfrac{g_3}{b_3},\dots,
    \dfrac{g_{2r+2}}{b_{2r+1}},-\dfrac{g_{2r+1}}{b_{2r+1}},0,\dots,0\right),
  \end{align*}
$g$ becomes
  \begin{align*}
g'= g'_\nu dx'^\nu=g_1dx'^1+\sum_{k=2r+3}^n g_k dx'^k.
  \end{align*}
  Here and in what follows,
  the case $G=dx^0dx^2$ is regarded as the $r=0$ case.

  Next, by $O_{n-2r-2}$ transformation
  preserving $(x'^{2r+3})^2+\dots+(x'^n)^2$,
  we can set
  \begin{align*}
    g=g_1dx^1+q dx^n,\
  \end{align*}
  where
      $q=\sqrt{(g_{2r+3})^2+\dots+(g_n)^2}$.

  The condition $\left[\xi^\sharp,\eta^\sharp\right]=\xi^\sharp$ leads to
  \begin{align*}
    (\G+I)f^\sharp=\F g^\sharp.
  \end{align*}
Since its covariant components are
  \begin{align*}
    (G_{\mu\nu}+\eta_{\mu\nu})f^\nu dx^\mu
   & =(-f^0+f^2)(dx^0+dx^2)+f^1dx^1\\
    &    +\sum_{i=1}^r
    \left[(f^{2i+1}+b_{2i+1}f^{2i+2})dx^{2i+1}+(f^{2i+2}-b_{2i+1}f^{2i+1})dx^{2i+2}
    \right]\\
    &+\sum_{k=2r+3}^n f^kdx^k,\\
    F_{\mu\nu}g^\nu dx^\mu
    &=g^1(dx^0-dx^2),
  \end{align*}
they should hold
  \begin{align*}
    f^2&=f^0,\\
    f^\mu&=0,~~~(\mu=1,3,4,\dots,n)\\
    g^1&=0.
  \end{align*}

  Moreover, by the translation
  \begin{align*}
    x'^\mu=x^\mu+f^0\delta^\mu_0,
  \end{align*}
we can set
  \begin{align*}
    f^0=0.
  \end{align*}
  
Hence the final forms of $f$ and $g$ are
  \begin{align*}
f&=0,\\
g&=qdx^n.
  \end{align*}
\end{proof}

The other possibility $F=0$ is covered by
the following Lemmas.

\begin{lem}\label{lem:F0-1}
  Let $\xi=f_\mu dx^\mu(\ne 0)$ be a translation Killing field.
  If $\xi$ is a timelike or spacelike  1-form,
  there is no Killing field $\eta$ that satisfies
  $\left[\xi^\sharp,\eta^\sharp\right]=\xi^\sharp$.
\end{lem}
\begin{proof}
  The equation  $\left[\xi^\sharp,\eta^\sharp\right]=\xi^\sharp$
  is equivalent to
  \begin{align*}
    \G f^\sharp=f^\sharp.
  \end{align*}
  By Corollary \ref{cor:null}, $f^\sharp$ should be a null vector.
\end{proof}

\begin{lem}\label{lem:F0-2}
  Let   $\xi=-dx^0+dx^1$ be a null translation Killing field.
  Then, every Killing field $\eta$
  that satisfies $\left[\xi^\sharp,\eta^\sharp\right]=\xi^\sharp$
  can be brought into the form
\begin{align*}
\eta=x^0dx^1-x^1dx^0+qdx^n,
\end{align*}
or
\begin{align*}
\eta&=x^0dx^1-x^1dx^0+\sum_{i=1}^{r}b_{2i}
(x^{2i}dx^{2i+1}-x^{2i+1}dx^{2i})+qdx^n,
&(b_2\ge b_4\ge\dots\ge b_{2r}>0)
\end{align*}
by a Poincar\'e transformation  that leaves $\xi$ invariant.
\end{lem}

\begin{proof}
For Killing fields
\begin{align*}
  \xi&=-dx^0+dx^1,\\
    \eta&=G_{\mu\nu}x^\mu dx^\nu+g_\nu dx^\nu,
  \end{align*}
the commutation relation $[\xi^\sharp,\eta^\sharp]=\xi^\sharp$
is equivalent to
  \begin{align*}
    (G_{\mu\nu}+\eta_{\mu\nu})\xi^\nu=G_{\mu 0}+G_{\mu 1}+\eta_{\mu 0}
    +\eta_{\mu 1}=0.
  \end{align*}
This leads to
  \begin{align*}
    G_{01}&=1,\\
    G_{k0}+G_{k1}&=0.&(k=2,3,\dots,n)
  \end{align*}
Hence the 2-form $G$ must be in the form
  \begin{align*}
    G=dx^0dx^1+\sum_{k=2}^n h_k(dx^0-dx^1)dx^k
    +\sum_{2\le i<j\le n}G_{ij}dx^idx^j.
  \end{align*}

By a Lorentz transformation 
  \begin{align*}
    x^0&=\left(1+\dfrac{\beta^2}{2}\right)x'^0-\dfrac{\beta^2}{2}x'^1
    -\sum_{k=2}^n\beta_kx'^k,\\
    x^1&=\dfrac{\beta^2}{2}x'^0+\left(1-\dfrac{\beta^2}{2}\right)x'^1
    -\sum_{k=2}^n\beta_kx'^k,\\
    x^k&=x'^k-\beta_k(x'^0-x'^1),\\
  \beta^2&=\sum_{k=2}(\beta_k)^2,
  \end{align*}
  that leaves  
the form of $\xi$ invariant,
$G$ becomes
  \begin{align*}
    G&=dx'^0dx'^1-\sum_{k=2}^n\left[\sum_{\ell =2}^n(\delta_{k\ell}-G_{k\ell})\beta_\ell-h_k\right]
    (dx'^0-dx'^1)dx'^k\\
    &+\sum_{2\le i<j\le n}G_{ij}dx'^idx'^j.
  \end{align*}
  Noting that  $(\delta_{k\ell}-G_{k\ell})$'s constitute
  a invertible matrix,
  the second term can be made zero
  by appropriately choosing $\beta_\ell$.

  Hence, by an $O_{n-1}$ transformation that preserves
$(x'^2)^2+\dots+(x'^n)^2$, it can be made
into the form
\begin{align*}
  G=dx^0dx^1,
\end{align*}
or
\begin{align*}
G&=dx^0dx^1+\sum_{i=1}^r b_{2i}dx^{2i}dx^{2i+1}.&(b_2\ge b_4\ge\dots\ge b_{2r}>0)
\end{align*}

Then, by a translation, we can made
\begin{align*}
g_0=g_1=\dots=g_{2r+1}=0,
\end{align*}
where the case $G=x^0dx^1$ corresponds to $r=0$. 

By further $O_{n-2r-1}$ transformation
that leaves $(x^{2r+2})^2+\dots+(x^n)^2$ invariant,
we can set
\begin{align*}
g_\nu dx^\nu=q dx^n.
\end{align*}
\end{proof}

Now we completely grasp all the possibilities
for the noncommutative pair of Killing filds from
Lemmas \ref{lem:FGF}, \ref{lem:FN}, \ref{lem:F0-1}, and
\ref{lem:F0-2}, which are summerized as  follows.

\begin{thm}\label{thm:killing-lie}
  Let $\xi$ and $\eta$ be Killing fields those are subject to
  $\left[\xi^\sharp,\eta^\sharp\right]=\xi^\sharp$.
  Then, by a Poincar\'e transformation,
  these can be brought into either of the forms
  \begin{align*}
    \xi&=(x^0-x^2)dx^1-x^1(dx^0-dx^2),\\
    \eta&=x^0dx^2-x^2dx^0
    +\sum_{i=1}^r b_{2i+1}(x^{2i+1}dx^{2i+2}-x^{2i+2}dx^{2i+1})+qdx^n,
  \end{align*}
  or
  \begin{align*}
    \xi&=-dx^0+dx^1,\\
    \eta&=x^0dx^1-x^1dx^0
    +\sum_{i=1}^r b_{2i}(x^{2i}dx^{2i+1}-x^{2i+1}dx^{2i})+qdx^n.
  \end{align*}
\end{thm}

\section{Conclusion}
We have shown that the equation of motion for
cohomogeneity-one Nambu-Goto strings in $\min{n}$ is completely integrable
for $n\ge 1$, which generalize the result of Koike et al. for $n=3$.
In order to enumerate the possible types of those equation of motion,
we have classified the canonical form of the Killing vector field
in $\min{n}$ under the action of the Poincar\'e group.
It has been shown that there are 4 types for $n=1$,
7 types for $n=2$, 11 types for $n=3$, 13 types for $n=4$,
and 14 types for $n\ge 5$ of the Killing vector fields (Theorem \ref{thm:killing}).

The equation of motion for a
cohomogeneity-one Nambu-Goto string in $\min{n}$ reduces to
the geodesic equation in the $n$-dimensional Killing orbit space.
The Hamilton-Jacobi equations for these
geodesic equations are generally non-separable.
Hence, at this point, the integrability issue of this system is not self-evident.

By a conformal tric \cite{KKI10}, the geodesic system in the curved orbit space
is converted  into the problem of
particle motions in $\min{n}$ with the
potential function determined by the amplitude of the Killing vector field.
Then, we have found that
all the canonical forms for the Killing vector fields defines a
simple Hamiltonian function for the particle system, enough to
be able to
show that the system is completely integrable.

We have also enumerate the possible forms of the basis of the
2-dimensional Lie algebra of Killing vector fields
subject to $[\xi^\sharp,\eta^\sharp]=\xi^\sharp$.
It has been found that there only 2-types for such pair of
Killing vector fields (Theorem \ref{thm:killing-lie}).
This result would be useful to classify the Nambu-Goto membranes
in $\min{n}$
those respects spacetime isometry
isomorphic to a 2-dimensional subgroup of the Poincar\'e group.


\section*{Acknowledgment}
We are grateful for
the stimulating discussions with Dr. T. Koike 
and Dr. Y. Morisawa during the Singularity Meeting in 2019.

\end{document}